\documentclass[11pt]{article}
\usepackage{amsmath,amsthm,amsfonts, amssymb}
\usepackage{multicol}
\usepackage[noend]{algorithmic}
\usepackage{algorithm,caption}
\usepackage{hyperref}
\usepackage{color}
\usepackage{epsfig}
\usepackage{epstopdf}

\DeclareMathOperator*{\argmax}{arg\,max}

\newtheorem{theorem}{Theorem}[section]
\newtheorem{definition}[theorem]{Definition}
\newtheorem{proposition}[theorem]{Proposition}
\newtheorem{lemma}[theorem]{Lemma}

\newcommand{\dist}{\mathrm{dist}}

\newcommand{\calI}{{\cal I}}
\newcommand{\calD}{{\cal D}}

\def\blinear#1{\overline{L(#1)}}

\newcommand{\M}[1]{\mathbf{#1}}
\newcommand{\V}[1]{\boldsymbol{\mathit{#1}}}

\newcommand{\PiPi}{\boldsymbol{\mathit{\Pi}}}

\newcommand{\ppi}{\boldsymbol{\mathit{\pi}}}

\def\LV{L(V)}
\def\bLV{\overline{L(V)}}

\def\B3CT{B$^3$CT}

\def\calC{\mathcal{C}}

\def\PPR{\mathbf{PPR}}
\def\PR#1{\mathbf{PageRank}_#1}

\def\cpr#1#2{p_{#1\rightarrow #2}}

\def\calR{\cal R}

\def\ppi{\boldsymbol{\mathit{\pi}}}

\def\GGamma{\boldsymbol{\mathit{\Gamma}}}

\def\00{\mathbf{0}}
\def\11{\mathbf{1}}

\def\cc{\mathbf{c}}

\def\dd{\mathbf{d}}

\def\ff{\mathbf{f}}
\def\hh{\mathbf{h}}

\def\pp{\mathbf{p}}

\def\ss{\mathbf{s}}
\def\uu{\mathbf{u}}
\def\ww{\mathbf{w}}
\def\vv{\mathbf{v}}
\def\xx{\mathbf{x}}
\def\yy{\mathbf{y}}
\def\zz{\mathbf{z}}

\def\DD{{\bf D}}
\def\FF{{\bf F}}
\def\II{{\bf I}}
\def\LL{{\bf \mathcal{L}}}
\def\LLL{{\bf \mathbf{L}}}

\def\MM{{\bf M}}

\def\UU{{\mathbf U}}

\def\WW{{\bf W}}

\newcommand{\overbar}[1]{\mkern 1.5mu\overline{\mkern-1.5mu#1\mkern-1.5mu}\mkern 1.5mu}

\def\Gbar{\overbar{G}}
\def\Ebar{\overbar{E}}
\def\WWbar{\overbar{\WW}}

\def\conductance{{\rm conductance}}
\def\pagerankconductance{\mbox{\rm PageRank-conductance}}
\def\pagerankutility{\mbox{\rm PageRank-utility}}
\def\pageranclusterability{\mbox{\rm PageRank-clusterability}}

\def\SCF{{\rm clusterability}}
\def\centrality{{\rm centrality}}

\def\bvec#1{{\mbox{\boldmath $#1$}}}

\def\calK{{\bf  \mathcal{K}}}
\def\calM{{\bf  \mathcal{M}}}
\def\calF{{\bf  \mathcal{F}}}

\def\expec#1#2{\mbox{\rm E}_{#1}\left[ #2 \right]}
\def\prob#1#2{\mbox{\rm Pr}_{#1}\left[ #2 \right]}
\def\nnz{\mbox{\rm nnz}}

\newdimen\pIR
\newcommand\R{{\rm I\kern\pIR R}}
\def\Reals#1{\mathbb{R}^{#1}}

\def\text#1{\mbox{#1}}

\begin{document}

\title{Network Essence: PageRank Completion and Centrality-Conforming Markov Chains}

\author{Shang-Hua Teng\thanks{Supported in
  part by a Simons Investigator Award from the Simons Foundation
  and by NSF grant CCF-1111270.
I would like to thank my Ph.D. student 
  Alana Shine and the three anonymous referees for valuable feedback.}\\ Computer Science and Mathematics\\  USC
}

\maketitle

\begin{abstract}
Ji\v{r}\'{i} Matou\v{s}ek (1963-2015) had many breakthrough contributions
 in mathematics and algorithm design.
His milestone results are not only profound  but also 
  elegant.
By going beyond the original objects --- 
  such as Euclidean spaces or linear programs ---
  Jirka found the {\em essence} of 
  the challenging mathematical/algorithmic problems 
  as well as beautiful solutions
  that were natural to him, but were surprising discoveries to the field.

In this short exploration article, 
  I will first share with readers my initial
  encounter with Jirka and discuss one of
  his fundamental geometric results from the early 1990s.
In the age of social and information networks, 
  I will then turn the discussion from 
  geometric structures to network structures,
  attempting to take a humble step towards the holy grail of network science,
  that is to understand the {\em network essence} 
  that underlies the observed sparse-and-multifaceted network data.
I will discuss a simple result which summarizes some basic 
   algebraic properties of {\em personalized PageRank matrices}.
Unlike the traditional transitive closure of binary relations, 
  the personalized PageRank matrices take ``accumulated Markovian closure'' 
  of network data. 
Some of these algebraic properties are known in various contexts.
But I hope featuring them together in a broader context
   will help to illustrate
  the desirable properties of this Markovian completion of networks,
  and motivate systematic developments of a network theory
  for understanding vast and ubiquitous multifaceted network data.
\end{abstract}


\section{Geometric Essence: To the Memory of Ji\v{r}\'{i} Matou\v{s}ek}

Like many in theoretical computer science and discrete mathematics,
  my own research has benefited from Jirka's deep insights, 
  especially into computational geometry \cite{JM}
  and linear programming \cite{JMLP}.
In fact, our paths accidentally crossed 
  in the final year of my Ph.D. program.
As a part of my 1991 CMU thesis \cite{TengThesis}, 
I obtained a result on the deterministic computation of 
  a geometric concept, called centerpoints, which led me to
  learn about one of Jirka's groundbreaking results during
  this time.

\subsection{Centerpoints}

The median is a widely-used concept for analyzing one-dimensional data,
  due to its statistical robustness
  and its natural algorithmic applications to divide-and-conquer.
In general, suppose $P = \{p_1,...,p_n\}$ is a set of $n$ real numbers.
For $\delta \in (0,1/2]$, we call $c\in \Reals{}$ 
  a {\em $\delta$-median} of $P$
  if $\max\left(|\{i : p_i< c\}|,|\{j : p_j> c\}|)\leq (1-\delta\right) n.$
A $\frac{1}{2}$-median of $P$ is known simply as a {\em median}.
Centerpoints are high-dimensional generalization of medians:

\begin{definition}[Centerpoints]\label{def:centerpoints}
Suppose $P = \{\pp_{1},\ldots,\pp_{n}\}$ is a point set 
  in $\Reals{d}$. 
For $\delta \in (0,1/2]$, a point $\cc\in \Reals{d}$
  is a {\em $\delta$-centerpoint} of $P$ if 
  for all unit vectors $\zz\in \Reals{d}$,
the projection $\zz^{T}\cc$ is a {\em $\delta$-median}
 of the projections, $\zz^T\cdot P = \{\zz^T\pp_{1},\ldots,\zz^T\pp_{n}\}$.
\end{definition}

Geometrically, every hyperplane $\hh$ in $\Reals{d}$ divides 
 the space into two open halfspaces,
 $\hh^+$ and $\hh^-$. 
Let the {\em splitting ratio} of $\hh$ over $P$,
denoted by $\delta_\hh(P)$, be:
\begin{eqnarray}
\delta_\hh(P) := \frac{\max\left(|\hh^+\cap P|,|\hh^-\cap P|\right)}{|P|}
\end{eqnarray}
Definition \ref{def:centerpoints}
  can be restated as:
  $\cc\in \R^d$ is a $\delta$-centerpoint of $P$
  if the splitting ratio of every hyperplane $\hh$ passing through $\cc$ 
  is at most $(1-\delta)$.
Centerpoints are fundamental to geometric divide-and-conquer 
   \cite{EdelsbrunnerBook}.
They are also strongly connected  to the concept of regression 
  depth introduced by Rousseeuw and 
  Hubert in robust statistics  \cite{RousseeuwHubert,AmentaBernEppstinTeng}.

We all know that every set of real numbers has a median.
Likewise --- and remarkably ---
  every point set in $d$-dimensional Euclidean space
  has a $\frac{1}{d+1}$-centerpoint
 \cite{DanzerFonluptKleeHelly}.
This mathematical result can be established by 
  Helly's  classical theorem from convex 
  geometry.\footnote{Helly's Theorem states: 
  Suppose $\calK$ is a family of at least $d+1$ convex sets
  in $\R^d$, and $\calK$ is finite or each member  
  of $\calK$ is compact.  
Then, if each $d+1$ members of $\calK$ 
  have a common point, there must be a point
  common to all members of $\calK$.}
Algorithmically,
  Vapnik-Chervonenkis' celebrated sampling theorem \cite{VC71} (more below)
  implies an efficient randomized algorithm  --- at least in theory --- 
  for computing a $(\frac{1}{d+1}-\epsilon)$-centerpoint.
This ``simple'' algorithm first takes a ``small''random sample, and then
  obtains its $\frac{1}{d+1}$-centerpoint via linear programming.
The complexity of this LP-based sampling algorithm is:
$$2^{O(d)} \left( \frac{d}{\epsilon^{2}}\cdot \log 
\frac{d}{\epsilon}\right)^d.$$

\subsection{Derandomization}

For my thesis,  I needed to compute centerpoints 
  in order to construct geometric separators \cite{MTTVJACM}
  for supporting finite-element simulation and
   parallel scientific computing \cite{MTTVSIAM}.
Because linear programming was too slow, I needed 
  a practical centerpoint algorithm to run large-scale 
  experiments \cite{GilbertMillerTeng}.
Because I was a theory student, I was also aiming
   for a theoretical algorithm to enrich my thesis.
For the latter, I focused on derandomization, which was then
  an active research area in theoretical computer science.
For centerpoint approximation without linear programming, 
  my advisor Gary Miller and I quickly obtained a simple 
  and practical algorithm\footnote{The construction started 
  as a heuristics, but it took a few more 
  brilliant collaborators  and years (after my graduation) to rigorously 
  analyzing its performance \cite{CEMST92}.} 
  based on Radon's classical theorem\footnote{Radon's Theorem states: Every point set $Q\subset\R^d$ with $|Q| \geq d + 2$ can be 
 partitioned into two subsets  $(Q_1,Q_2)$ such that 
 the convex hulls of $Q_1$ and $Q_2$ have a common point.
}
  \cite{DanzerFonluptKleeHelly}.
But for derandomization, it took me more than a year to finally design 
  a deterministic linear-time algorithm 
  for computing $(\frac{1}{d+1}-\epsilon)$-centerpoints 
  in any fixed dimensions.
It happened in the Spring of 1991, my last semester at CMU.
Gary then invited me to accompany him 
  for a month-long visit, starting at the spring break of 1991,
  at the International Computer Science Institute (ICSI),
  located near the U.C. Berkeley campus.
During the  California visit, 
I ran into Leo Guibas, 
  one of the pioneers of computational geometry.

After I told Leo about my progress on  Radon-Tverberg 
  decomposition \cite{Tverberg66} and centerpoint computation,
  he mentioned to me a paper by Jirka \cite{JM},
  which was just accepted to the
  {\em ACM Symposium on Theory of Computing} (STOC 1991) 
  ---  before my solution  ---  that beautifully solved the 
  sampling problem for a broad class of
  computational geometry and statistical learning problems.
Jirka's result --- see Theorem \ref{theo:JM} below --- 
  includes the approximation of centerpoints as a simple special case.
Although our approaches had some ideas in common, 
  I instantly knew that this mathematician ---
  who I later learned was just a year older than me ---
  was masterful and brilliant. 
I shortened that section of my thesis
 by referring readers to Jirka's paper \cite{JM},
  and only included the scheme I had
  that was in common with his bigger result.

\begin{figure}[h]
\centering
 \includegraphics[width=0.80\linewidth]{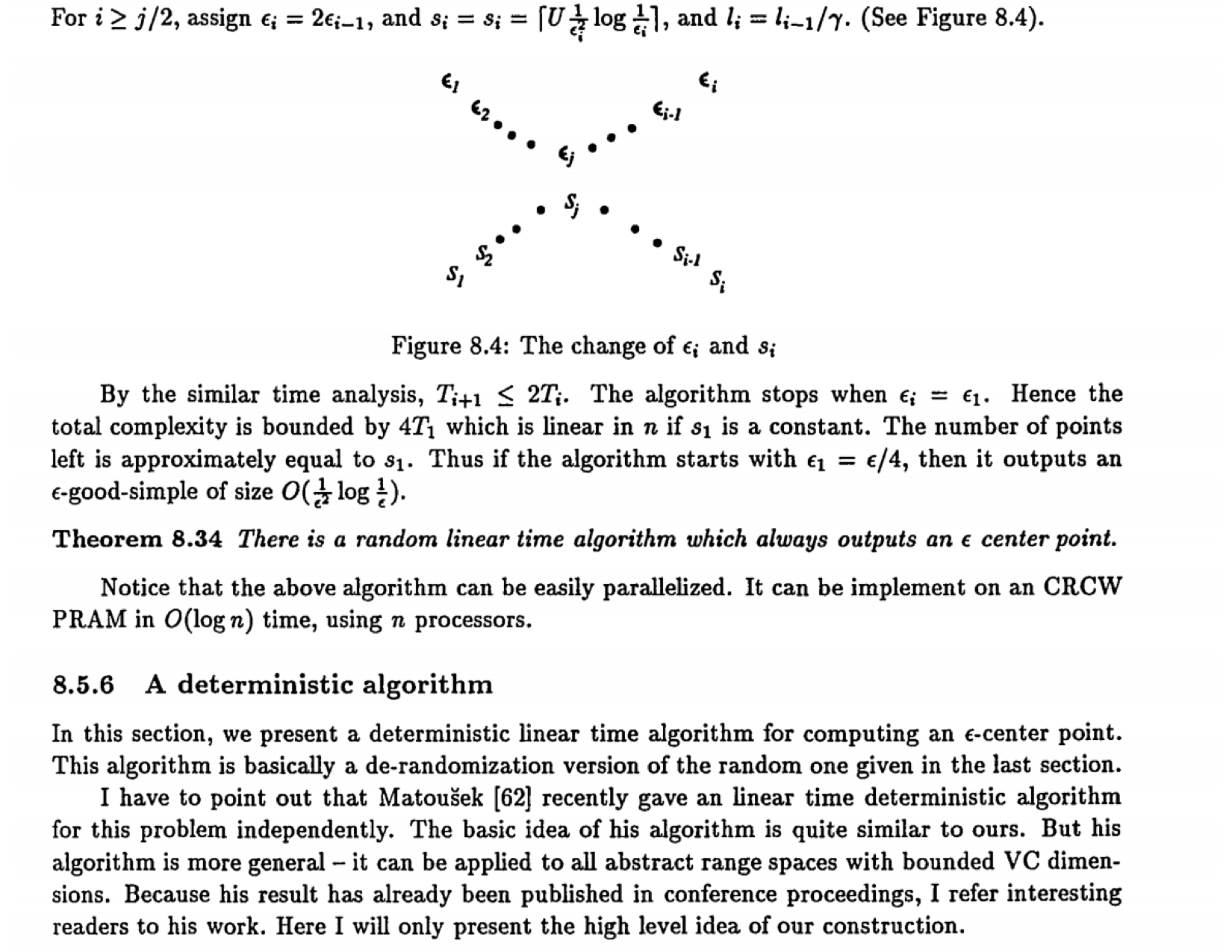}
 \caption{Page 66 (Chapter 8) of my thesis}
        \label{fig:epsilon}
\end{figure}

\subsection{Matou\v{s}ek's Theorem: The Essence of Dimensionality}
Mathematically, a {\em range space} $\Sigma$ is a pair $( X , \calR )$, 
  where $X$ is a finite or infinite set,
  and $\calR$ is a finite or infinite family of subsets of $X$. 
Each $H\in \calR$ can be viewed as a classifier of $X$, with elements
  in $X \cap H$ as its positive instances.
For example, $\Reals{d}$ and its halfspaces form a range space,
  so do $\Reals{d}$ and its $L_p$-balls, for any $p> 0$, 
  as well as $V$ and the set of all cliques in a graph $G = (V,E)$.
Range spaces greatly extend the concept of {\em linear separators}.

An important technique in statistical machine learning and computational geometry
  is sampling.
For range spaces, we can measure the quality of a sample as the following:
\begin{definition}[$\epsilon$-samples]\label{def:sample}
Let $\Sigma = (X,\calR)$ be an $n$-point range space.
A subset $S \subseteq X$ is an {\em $\epsilon$-sample} 
  or {\em $\epsilon$-approximation} for $\Sigma$ if for all $H\in \calR$: 
\begin{eqnarray}
\left| \frac{|H\cap S|}{|S|} - \frac{|H\cap X|}{|X|}\right| \leq \epsilon
\end{eqnarray}
\end{definition}

For each $S \subseteq X$, the set of distinct classifiers
   that $\calR$ can define is ${\calR}(S) = \{H\cap S: H\in {\calR}\}$.
For any $m\leq |X|$, let the {\em  shatter function} for $\Sigma$ be:
\begin{eqnarray}
\pi_{\calR}(m) = \sup_{S\subseteq X, |S| = m} \left|{\calR}(S)\right|
\end{eqnarray}

\begin{theorem}[Deterministic Sampling -- Matou\v{s}ek]\label{theo:JM}
Let $\Sigma = (X, \calR)$
   be an $n$-point range space 
    with the shatter function satisfying 
$\pi_{\calR}(m) = O(m^d)$ ($d\geq 1$ a constant). 
Having a subspace oracle for $\Sigma$, and given a parameter $r$, we can 
  deterministically compute a $(1/r)$-approximation of size
  $O(dr^2\log r)$ for $\Sigma$, in time $O(n(r^2\log r)^d)$.
\end{theorem}

Matou\v{s}ek's sampling theorem goes beyond traditional geometry 
  and completely derandomizes the theory of Vapnik-Chervonenkis \cite{VC71}.

\begin{theorem}[Vapnik and Chervonenkis] \label{theo:VC}
There exists a constant $c$ such that 
  for any finite range space $\Sigma =(X,\calR)$ 
  and $\epsilon, \delta \in(0,1)$, if $S$ is a set of 
$ c\cdot \frac{d}{\epsilon^2}\left(\log\frac{d}{\epsilon\delta})\right)$
 uniform and independent samples
 from $X$, where $d = \mbox{\rm VC}(\Sigma)$, then:
$\prob{}{\mbox{$S$ is an $\epsilon$-sample for $\Sigma$}} \geq 1 -\delta$
\end{theorem}

Matou\v{s}ek's deterministic algorithm
  can be applied to geometric classifiers 
  as well as any classifier --- known as a
 {\em concept space} --- 
  that arises in statistical learning theory \cite{VapnikBook}.
The concept of range space has also provided a powerful tool for 
  capturing geometric structures, and played a profound role --- both in theory
  and in practice --- for data clustering \cite{Feldmancoreset} and
 geometric approximation \cite{AgarwalCoreset}.
The beauty of 
  Vapnik-Chervonenkis' theory and Matou\v{s}ek's sampling theorem
  lies in the {\em essence of dimensionality},
  which is generalized from geometric spaces to abstract range spaces.
In Euclidean geometry, the dimensionality comes naturally
  to many of us.
For abstract range spaces, the growth of the {\em  shatter functions}
  is more intrinsic!
If $\pi_{\calR}(m) = 2^m$, then there exists a set $S\subseteq X$
  of $m$ elements that is {\em shattered}, i.e., 
  for any subset $T$ of $S \subseteq X$, there exists
  $H\in \calR$ such that $T = H \cap S$.
In other words, we can use $\calR$ to build classifiers for
   all subsets of $S$.
There is a beautiful {\em dichotomy} of polynomial and exponential
  complexity within the concept of shattering:
\begin{itemize}
\item   either $X$ has a subset $S \subseteq X$ of size $m$ that can be shattered
  by $\calR$,
\item  or for any $U \subseteq X$, $|U| \geq m$, 
  $|\{H\cap U: H\in \calR\}|$ is polynomial in $|U|$.
\end{itemize}
The latter case implies that $\calR$ can only be used to build
 a polynomial number of classifiers for $U$.
The celebrated {\em  VC-dimension} of range space $\Sigma = (X,\calR)$,
denoted by 
$\mbox{\rm VC}(\Sigma)$, is defined as: 
$$\mbox{\rm VC}(\Sigma) := \argmax \{m:  \pi_{\calR}(m) = 2^m\}.$$
This polynomial-exponential dichotomy is 
  established by the following
  Sauer's lemma.\footnote{This lemma is also
known as Perles-Sauer-Shelah's lemma.}

\begin{lemma}[Sauer]
For any range space  $\Sigma = (X,\calR)$ and $\forall m> \mbox{\rm VC}(\Sigma)$, 
$\pi_{\calR}(m) \leq \sum_{k=0}^{\mbox{\small\rm VC}(\Sigma)}
{m \choose k}$.
\end{lemma}

Sauer's lemma extends the following well-known fact
  of Euclidean geometry:
   any set of   $m$ hyperplanes in $\Reals{d}$ divides the space into 
  at most $O(m^d)$ convex cells.
By the point-hyperplane duality, any set of $m$ points can be divided
  into at $O(m^d)$ subsets by halfspaces.

Although  my construction of $\epsilon$-samples in $\Reals{d}$ 
  was good enough for designing linear-time 
  centerpoint approximation algorithm in fixed dimensions,
  it did not immediately generalize to arbitrary range spaces, 
   because it was tailored to the
   geometric properties of Euclidean spaces.

By addressing abstract range spaces, 
  Jirka resolved the intrinsic algorithmic problem at the heart 
  of  Vapnik-Chervonenkis' sampling theory.
Like Theorem \ref{theo:JM},
  many of Jirka's other landmark and breakthrough
  results are elegant, insightful, and fundamental.
By going beyond the original objects
  --- such as Euclidean spaces or  linear programs \cite{JMLP} ---
  Jirka usually went directly to the {\em essence} of 
  the challenging problems to come up with beautiful solutions
  that were natural to him but remarkable 
  to the field.


\section{Backgrounds: Understanding Multifaceted Network Data}

To analyze the structures of social and information networks
  in the age of Big Data, we need to overcome
  various conceptual and algorithmic challenges both in understanding
  network data and in formulating solution concepts.
For both, we need to capture the {\em network essence}.

\subsection{The Graph Model  --- A Basic Network Facet}

At the most basic level, 
  a network can be modeled as a graph $G = (V,E)$,
  which characterizes the structure of the network
  in terms of:
\begin{itemize}
\item {\bf nodes}: for example, Webpages, Internet routers,
  scholarly articles, people,  random variables, or counties
\item {\bf edges}: for example, links, connections, citations, friends, conditional dependencies,  or voting similarities
\end{itemize}
In general, nodes in many real-world networks may not be 
  ``homogeneous'' \cite{MultiNodeGraph}, 
  as they may have some additional features, specifying 
  the {\em types} or {\em states} of the node elements.
Similarly, edges may have additional features, specifying
  the {\em levels} and/or {\em types} of 
  pairwise interactions, associations, or affinities.

Networks with  ``homogeneous'' types of nodes and edges are closest to
  the combinatorial structures studied under traditional graph theory, which 
  considers both weighted or unweighted graphs.
Three basic classes of weighted graphs 
  often appear in applications.  
The first class consists of {\em distance networks},
  where each edge $e\in E$ is assigned a number $l_e\geq 0$,
  representing the {\em length} of edge $e$.
The second class consists of {\em affinity networks},
  where each edge $(u,v) \in E$ is assigned a  weight $w_{u,v} \geq 0$,
  specifying $u$'s {\em affinity weight} towards $v$.
The third class consists of {\em probabilistic networks},
  where each (directed) edge $(u,v) \in E$ is assigned a probability 
  $p_{u,v} \geq 0$, modeling how a random process  connects $u$ to $v$.
It is usually more natural to view
    maps or the Internet  as distance networks,
    social networks as affinity networks, and 
    Markov processes as  probabilistic networks.
Depending on applications, a graph may be directed or undirected.
Examples of directed networks include:
 the Web, Twitter, the citation graphs of scholarly publications, 
  and Markov processes.
Meanwhile, Facebook ``friends'' or collaboration networks 
  are examples of undirected graphs.

In this article, we will first focus on affinity networks.
An affinity network with $n$ nodes can be mathematically represented
   as a weighted graph  $G = (V,E,\WW)$.
Unless otherwise stated, we assume 
  $V = [n]$ and $\WW$ is an $n\times n$ non-negative matrix (for example from 
$[0,1]^{n\times n}$).
We will follow the convention that for $i\neq j$, $w_{i,j} = 0$,
  if and only  if, $(i,j)\not\in E$.
If $\WW$ is a symmetric matrix, then we say $G$  is {\em undirected}.
If $w_{i,j} \in \{0,1\}$, $\forall i,j\in V$,  
  then we say $G$ is {\em unweighted}. 

Although they do not always fit, 
  three popular data models for defining pairwise affinity weights 
  are the metric model, feature model, and statistical model.
The first assumes that an underlying metric space,
  $\calM = \left(V,\dist\right)$,
   impacts the interactions among nodes in a network.
The affinities  between nodes may 
   then be determined by their {\em distances} from
  the underlying metric space:
The closer two elements are, the higher their affinity
  becomes, and the more interactions they have.  
A standard way to define affinity  weights for  $u\neq v$ is: 
$w_{u,v} = \dist(u,v)^{-\alpha}$,
  for some $\alpha > 0$.
The second assumes   that there exists an underlying ``feature'' space,
  $\calF = \left(V, \FF\right)$,
  that impacts the interactions among nodes in a network.
This is a widely-used alternative data model for information networks.
In a $d$-dimensional feature space, $\FF$ is an $n\times d$ matrix, where
  $f_{u,i} \in  \Reals{+}\cup\{0\}$ denotes $u$'s {\em quality score} 
   with respect the $i^{th}$ feature.  
Let $\ff_u$ denote the $u^{th}$ row of $\FF$, i.e., 
  the {\em feature vector} of node $u$.
The affinity weights  $w_{u,v}$ between two nodes $u$ and $v$ may 
  then be determined by the {\em correlation} between
  their features:
$w_{u,v} \sim \left(\ff_u^T\cdot \ff_v\right) = \sum_{i=1}^d f_{u,i}\cdot f_{v,i}.$
The third assumes that there exists an underlying 
  statistical space (such as a stochastic block model, Markov process, or (Gaussian) random field)
  that impacts the pairwise interactions.
The higher the dependency between two elements is, the higher their 
 strength of tie is.

If one thinks that the meaning of weighted networks is complex, the 
  real-world network data is far more complex and diverse.
We will have more discussions in Section \ref{sec:RichnessOther} and
 Section \ref{sec:RW}.

\subsection{Sparsity and Underlying Models} 

A basic challenge in network analysis is that
  real network data that we observe is only a reflection of
  underlying network models.
Thus, like machine learning tasks which have to work 
  with samples from an unknown underlying distribution,
  network analysis tasks typically work 
  with observed network data, 
  which is usually different from the underlying network model.
As argued in 
  \cite{B3CT,HXNetworkCompletion,KLNetworkCompletion,TengScalable}, 
  a real-world  social and information network may be viewed as an
   observed network, induced by a ``complete-information''  underlying 
   preference/affinity/statistical/geometric/feature/economical model.
However, these observed 
   networks are typically sparse with many missing links.

\begin{quote}
{\em 
For studying network phenomena, it is crucial to mathematically
   understand underlying network models, while 
  algorithmically work efficiently with sparse observed data. 
Thus, developing systematic approaches to uncover or 
  capture the 
  underlying network model --- or the network essence --- is 
  a central and challenging mathematical task  in network analysis.
}
\end{quote}
Implicitly or explicitly, underlying network models are the ultimate
  guide for understanding network phenomena, and
  for inferring missing network data, and distinguishing missing links
  from absent links.
To study basic network concepts, we also need to simultaneously
  understand the observed and underlying networks.
Some network concepts, such as {\em centrality}, capture 
 various aspects of ``dimension reduction'' of network data.
Others characterizations, 
  such as  {\em clusterability}  and {\em community classification},
  are more naturally expressed in 
  a space with  dimension higher than that of the observed networks. 

Schematically, centrality assigns a numerical score or ranking
  to each node, which measures the {\em importance}
  or {\em significance} of each node in a network 
  \cite{PageRank,NewmanBook,Bonacich1987power,Freeman,FreemanBetweenness,BorgattiCentrality,BorgattiEverett,EverettBorgattiGroup,BonacichGroupCentrality,Faust,Katz,BavelasCloseness,SabidussiCloseness,PercolationCentrality,Donetti2005Entangled,ShapleyValueForCentrality1,ShapleyValueForCentrality2}.
Mathematically, a numerical centrality measure is a mapping from
  a network $G=(V,E,\WW)$ to a $|V|$-dimensional real vector:
\begin{eqnarray}
\left[\ \centrality_{\WW}(v)\ \right]_{v\in V} \in \R^{|V|}
\end{eqnarray}

For example, a widely used centrality measure is the 
  PageRank centrality.
Suppose $G = (V,E,\WW)$ is a weighted directed graph.
The {\em PageRank centrality} uses an
   additional parameter $\alpha\in (0,1)$ ---
   known as the {\em restart constant} ---
   to define a finite Markov process whose transition rule
  --- for any node $v\in V$ --- is the following:
\begin{itemize}
\item  with probability $\alpha$, restart at a random node in $V$, and
\item  with probability $(1-\alpha)$,
   move to a neighbor of $v$,  chosen randomly with probability
   proportional to edge weights out of $v$.
\end{itemize}
Then, the  {\em  PageRank centrality} (with restart constant $\alpha$)
  of any $v\in V$  is proportional to
 $v$'s stationary probability in this Markov chain.

In contrast, clusterability
  assigns a numerical score or ranking to each subset of nodes,
  which measures the {\em coherence}
  of each group in a network \cite{NewmanBook,LovaszSimonovits,TengScalable}.
Mathematically, a numerical clusterability measure
 is a mapping from a network $G = (V,E,\WW)$
 to a $2^{|V|}$-dimensional real vector:
\begin{eqnarray}
\left[\ \SCF_{\WW}(S)\ \right]_{S\subseteq V} \in [0,1]^{2^{|V|}}
\end{eqnarray}
An example of clusterability measure is
  {\em conductance} 
  \cite{LovaszSimonovits}.\footnote{The {\em conductance} of 
  a group $S\subset V$ is the ratio of its {\em external connection} to its
  {\em total connection} in $G$.}
Similarly, a community-characterization rule \cite{BCMTITCS} 
 is a mapping from a network $G = (V,E,\WW)$
 to a $2^{|V|}$-dimensional Boolean vector:
\begin{eqnarray}
\left[\ \calC_{\WW}(S)\ \right]_{S\subseteq V} \in \{0,1\}^{2^{|V|}}
\end{eqnarray}
indicating whether or not each group $S\subseteq V$ is a community
  in $G$.
Clusterability and community-identification rules
  have much higher dimensionality than centrality.
To a certain degree, they can be viewed as a 
  ``complete-information'' model of the observed network.
Thus again:
\begin{quote}
{\em  Explicitly or implicitly,
  the formulations of these network 
  concepts are  mathematical processes
  of uncovering or capturing underlying network models.
}
\end{quote}


\subsection{Multifaceted Network Data: Beyond Graph-Based Network Models}
\label{sec:RichnessOther}
Another basic challenge in network analysis
  is that 
  real-world network  data is 
  much richer than the graph-theoretical representations.
For example, social networks are more than weighted graphs. 
Likewise, the Web and Twitter are not just  directed graphs.  
In general, network interactions and phenomena 
  --- such as social influence \cite{Kempe03} or electoral behavior
 \cite{eulau1980columbia}, ---    are more complex than 
  what can be captured by nodes and edges.
The network interactions are often the result of the interplay between 
  dynamic mathematical processes and static underlying graph structures 
  \cite{GhoshLermanTengYan,ChenTeng}.

\subsubsection*{{\sc Diverse Network Models}}

The richness of network data and diversity of network concepts
  encourage us to consider network 
  models beyond graphs \cite{TengScalable}.
For example, 
  each clusterability measure
  $\left[\SCF_{\WW}(S)\right]_{S\subseteq V}$ of a weighted graph 
  $G=(V,E,\WW)$ explicitly 
  defines a complete-information, weighted hyper-network:

\begin{definition}[Cooperative Model: Weighted Hypergraphs]
A weighted hypergraph over $V$  is given by $H = (V,E,\V{\tau})$
  where $E\subseteq 2^V$ is a set of hyper-edges and 
  $\V{\tau}: E \rightarrow \Reals{}$ is a function 
  that assigns weights to hyper-edges.
$H$ is a complete-information cooperative networks if
  $E= 2^V$.
\end{definition}

We refer to weighted hypergraphs as {\em cooperative networks} 
  because they are the central subjects 
  in classical {\em cooperative game theory}, but under
   a different name \cite{ShapleyValue}.
An $n$-person {\em cooperative game} over $V = [n]$
  is specified by a {\em characteristic function}
  $\V{\tau}: 2^V \rightarrow \R$,
  where for any coalition $S\subseteq V$,
  $\V{\tau}(S)$ denotes the {\em cooperative utility}  of  $S$.

Cooperative networks are generalization of undirected weighted graphs.
One can also generalize directed networks,
  which specify directed node-node interactions.
The first one below explicitly captures
  node-group interactions,
  while the second one captures group-group interactions.

\begin{definition}[Incentive Model]\label{def:incentiveNetwork}
An {\em incentive network} over $V$ is a pair $U = (V,\bvec{u})$.
For each $s \in V$, 
  $u_s: 2^{V\setminus \{s\}} \rightarrow \Reals{}$ 
  specifies $s$'s {\em incentive utility}
  over subsets of $V\setminus \{s\}$.
In other words, there are $|S|$ utility values, 
  $\{u_s(S \setminus \{s\}\}_{s\in S}$,
  associated with each group $S\subseteq V$ in
  the incentive network. 
For each $s\in S$, the {\em value}
   of its interaction with the rest of the group
  $S \setminus \{s\}$ is explicitly 
  defined as $u_s(S \setminus \{s\})$.
\end{definition}

\begin{definition}[Powerset Model]\label{def:influenceNetwork}
A {\em powerset network} over $V$ is a weighted directed
  network on the powersets of $V$.
In other words, a powerset network $P = (V,\V{\theta})$
 is specified by a function
  $\V{\theta}: 2^V\times 2^V \rightarrow \Reals{}$.
\end{definition}

For example --- as pointed in \cite{Kempe03,ChenTeng} ---
  a social-influence instance fundamentally defines
  a powerset network.
Recall that a social-influence instance $\calI$
   is specified by a directed graph $G=(V,E)$ 
   and an influence model $\calD$
  \cite{Kempe03,DomingosRichardson,RichardsonDomingos}, where
$G$ defines the graph structure of the social network
  and $\calD$ defines a stochastic process that characterizes
  how nodes in each {\em seed set} $S\subseteq V$  
  {\em collectively influence}  
  other nodes using the edge structures of $G$ \cite{Kempe03}.
A popular influence model is 
  {\em independent cascade} (IC)\footnote{
In the classical IC model, each directed edge $(u,v) \in E$
  has an influence probability $p_{u,v}\in [0,1]$.
The probability profile defines a discrete-time influence process
  when  given a seed set $S$:
At time $0$,  nodes in $S$ are activated  while other nodes are inactive.
At time $t \ge 1$,
   for any node $u$ activated at time $t-1$, it has one
   chance to activate each of its inactive out-neighbor
   $v$ with an independent probability of $p_{u,v}$.
When there is no more activation,
  this stochastic process ends with a random set 
  of nodes activated during the process.
} \cite{Kempe03}.

Mathematically, the influence process $\calD$ and the network structure $G$
  together define a probability distribution 
  $\V{P}_{G,\calD}: 2^V\times 2^V\rightarrow [0,1]$:
For each $T \in 2^V$, 
 $\V{P}_{G,\calD}[S,T]$ 
 specifies the probability that $T$ is the final activated set 
  when $S$ cascades its influence through the network $G$.
Thus, $P_{\calI}=(V,\V{P}_{G,\calD})$ defines a natural powerset network,
  which can be viewed as the {\em underlying network}
  induced by the interplay between the {\em static} network structure $G$ and {\em dynamic} influence process $\calD$.

An important quality measure of $S$ in this process is $S$'s  
  {\em influence spread}  \cite{Kempe03}.
It can be defined from the powerset model $P_{\calI}=(V,\V{P}_{G,\calD})$ 
  as following:
\[ \V{\sigma}_{G,\calD}(S) = \sum_{T\subseteq V} |T|\cdot \V{P}_{G,\calD}[S,T].\]
Thus, $(V,\V{\sigma}_{G,\calD})$ also defines a natural 
  cooperative network \cite{ChenTeng}.

In many applications and studies, {\em ordinal network models}
   rather than {\em cardinal network models}
  are used to capture the preferences among nodes.
Two classical applications of preference frameworks
  are voting 
   \cite{ArrowBook} and stable marriage/coalition formation
  \cite{GaleShapley,RothSM,stableMarriage,BramsCoalition}.
A 
  modern use of preference models is
  the {\em Border Gateway  Protocol} (BGP) 
  for network routing 
  between autonomous  Internet systems  \cite{BGP,BGPSurvay}.

In a recent axiomatic study of community identification in social networks, 
  Borgs {\em et al.} \cite{BCMTITCS,B3CT}
  considered the following {\em abstract} social/information
   network framework. 
Below, for a non-empty finite set $V$, let $\LV$
  denote the set of all {\em linear orders} on $V$.

\begin{definition}[Preference Model]\label{def:preferenceNetwork}
A {\em preference network} over $V$ is a pair $A = (V,\Pi)$, where
 $\Pi =\{\V{\pi}_u\}_{u\in V}\in \LV^{|V|}$ 
 is a {\em preference profile} in which $\V{\pi}_u$ 
 specifies  $u$'s individual preference.
\end{definition}

\subsubsection*{{\sc Understanding Network Facets and Network Concepts}}

Each network model enables us to focus on 
   different facets of network data.
For example, 
  the powerset model offers the most natural framework for capturing
  the underlying interplay between influence processes 
  and network structures.
The cooperative model matches the explicit representation 
  of clusterability, group utilities, and influence spreads.
While traditional graph-based network data often consists solely of 
  pairwise interactions, affinities, or associations,  
  a community is formed by a group of individuals.
Thus, the basic question for community identification is to understand
  ``how do individual preferences (affinities/associations) 
  result in group preferences or community coherence?'' \cite{BCMTITCS}
The preference model highlights
  the fundamental aspect of community characterization.
The preference model is also natural for addressing the
  question of summarizing individual preferences 
  into one collective preference, which is 
  fundamental in the formulation of network centrality \cite{TengScalable}.
Thus, studying network models beyond graphs 
   helps to broaden our understanding of social/information networks.

Several these network models, as defined above, are
  highly theoretical models.
Their complete-information profiles have exponential dimensionality
  in $|V|$.
To use them as underlying models in network analysis,
  succinct representations should be constructed to efficiently
  capture observed network data.
For example, both the conductance clusterability measure
  and the social-influence powerset network are succinctly defined.
Characterizing network concepts in these models and 
  effectively applying them 
  to understanding real network data
  are promising and fundamentally challenging 
  research directions in network science.


\section{PageRank Completion}  \label{Sec:NE}

Network analysis is a task to capture the {\em essence} of the observed networks.
For example, graph embedding \cite{LLRGraphEmbedding,TengScalable} 
  can be viewed as a process to 
  identify the geometric essence of networks.
Similarly, network completion 
  \cite{HXNetworkCompletion,KLNetworkCompletion,MasrourCompletion},   
  graphon estimation 
  \cite{BorgsChayesSmith,graphonEstimation}, and
  community recovering in hidden stochastic block models \cite{AbbeSandonII}
  can be viewed as processes to distill the statistical essence
  of networks.
All these approaches build {\em constructive maps}
  from observed sparse
  graphs to underlying complete-information models.
In this section, we study the following basic question:

\begin{quote}
{\em 
Given an observed sparse affinity network $G= (V,E,\WW)$, 
 can we construct a complete-information affinity network 
  that is consistent with $G$?
}
\end{quote}

This question is simpler than but relevant to matrix and network completion 
  \cite{HXNetworkCompletion,KLNetworkCompletion},
  which aims to infer the missing data from sparse, observed network data.
Like matrix/network completion,
  this problem is mathematically an inverse problem.
Conceptually, we need to formulate the meaning of 
``a complete-information affinity network consistent with $G$.''

Our study is also partially motivated by the following question asked in
\cite{AltmanTennenholtzPersonalized,B3CT}, aiming to deriving personalized 
  ranking information from graph-based network data:
\begin{quote}
{\em 
Given a sparse affinity network $G= (V,E,\WW)$, 
 how should we construct a complete-information preference model 
 that best captures the underlying individual preferences
 from network data given by  $G$?
}
\end{quote}


We will prove 
  the following basic structural result\footnote{See Theorem \ref{theo:PageRankEssence} for the 
  precise statement.}:
Every connected, undirected, weighted graph  $G= (V,E,\WW)$
 has an undirected and weighted
   graph $\Gbar= (V,\Ebar,\WWbar)$, such that:
\begin{itemize}
\item {\bf Complete Information}: 
   $\Ebar$ forms  a complete graph with $|V|$ self-loops.
\item {\bf Degree and Stationary Preserving}:  
 $\WW\cdot \11 = \WWbar\cdot \11$. Thus, 
  the random-walk Markov chains on $G$ and on $\Gbar$ 
  have the same stationary distribution.
\item {\bf PageRank Conforming}: 
The transition matrix $\MM_{\WWbar}$ of the random-walk 
  Markov chain on  $\Gbar$ is {\em conformal} to the PageRank of $G$,
  that is, $ \MM_{\WWbar}^T\cdot \11$ is proportional to the 
  PageRank centrality  of $G$
\item {\bf Spectral Approximation}: 
  $G$ and $\Gbar$ are spectrally similar.
\end{itemize}

In the last condition, the similarity between $G$ and $\Gbar$ 
  is measured by the following 
  notion of {\em spectral similarity} 
  \cite{SpielmanTengSpectralSparsification}:
\begin{definition}[Spectral Similarity of Networks]
Suppose $G = (V,E,\WW)$ and $\Gbar = (V,\Ebar,\WWbar)$ 
  are two weighted undirected graphs  over the same set $V$ of $n$ nodes.
Let  $\LLL_{\WW} = \DD_{\WW} - \WW$ and 
  $\LLL_{\WWbar} = \DD_{\WWbar}- \WWbar$ 
  be  the {\em Laplacian matrices}, respectively, of these two 
  graphs.
Then, for $\sigma \geq 1$, we say 
   $G$ and $\Gbar$ are {\em $\sigma$-spectrally similar}
  if:
\begin{equation}\label{eqn:spectralSimilarity}
\forall \xx \in \Reals{n}, \quad  \frac{1}{\sigma}\cdot \xx^{T} \LLL_{\WWbar} \xx \leq \xx^{T} \LLL_{\WW} \xx \leq 
  \sigma\cdot \xx^{T} \LLL_{\WWbar} \xx
\end{equation}
\end{definition}
Many graph-theoretical measures, such as flows, cuts, 
  conductances, effective resistances, 
  are approximately preserved 
  by spectral similarity \cite{SpielmanTengSpectralSparsification,SparsificationCACM}.
We refer to $\Gbar= (V,\Ebar,\WWbar)$ 
  as the {\em PageRank essence} or {\em PageRank completion} 
  of $G = (V,E,\WW)$.

\subsection{The Personalized PageRank Matrix}

$\Gbar= (V,\Ebar,\WWbar)$ stated above is derived from 
  a well-known structure in network analysis, the personalized 
  PageRank matrix  of a network \cite{PageRankContribution,TengScalable}.

\subsubsection*{{\sc Personalized PageRanks}}

Generalizing the Markov process of PageRank, 
  Haveliwala \cite{Haveliwala03} introduced 
   personalized PageRanks.
Suppose $G = (V,E,\WW)$ is a weighted directed graph 
  and $\alpha > 0$ is a restart parameter.
For any distribution $\ss$ over $V$,
 consider the following Markov process,
   whose transition rule ---  for  any $v\in V$ --- is the following:
\begin{itemize}
\item with probability $\alpha$, restart at a random node in $V$
  according to distribution $\ss$, and
\item  with probability $(1-\alpha)$,
   move to a neighbor  of $v$, 
   chosen randomly with probability
   proportional to edge weights out of $v$.
\end{itemize}
Then, the {\em PageRank with respect to the
   starting vector $\ss$}, 
  denoted by $\pp_{\ss}$,
  is the stationary  distribution of this Markov chain.

Let $d^{out}_u = \sum_{v\in V} w_{u,v}$
    denotes the {\em out-degree}
  of $u\in V$ in $G$.
Then,  $\pp_{\ss}$ is the solution to the following equation:
\begin{eqnarray}\label{PPR} 
\pp_{\ss} = \alpha \cdot \ss + (1-\alpha)
  \cdot  \WW^T\cdot\left(\DD_{\WW}^{out}\right)^{-1}\cdot\pp_{\ss}
\end{eqnarray}
where  $\DD_{\WW}^{out} = {\rm diag}([d_1^{out},...,d_n^{out}])$ is
  the diagonal matrix of
  out degrees.
Let $\11_u$ denote the $n$-dimensional  vector
  whose $u^{th}$ location is 1 and all other entries in 
  $\11_u$ are zeros.
Haveliwala \cite{Haveliwala03} referred to $\pp_u := \pp_{\11_u}$ 
  as the {\em personalized PageRank} of $u\in V$ in $G$.
Personalized PageRank is asymmetric, and hence to emphasize this fact,
 we express $\pp_u$ as:
\[\pp_u = \left(\cpr{u}{1}, \cdots, \cpr{u}{n}\right)^T.\]
Then $\{\pp_u\}_{u\in V}$ --- the personalized PageRank profile ---
   defines  the following  matrix:
\begin{definition}[Personalized PageRank Matrix]\label{def:PPRM}
The {\em personalized PageRank matrix} of an
   $n$-node weighted graph  $G = (V,E, \WW)$ and restart constant $\alpha > 0$ is:
\begin{eqnarray}\label{eqn:PPR}
\PPR_{{\WW,\alpha}} = 
\left[\pp_{1},...,\pp_{n}\right]^T
=
\left[
\begin{array}{ccc} 
   \cpr{1}{1} & \cdots & \cpr{1}{n}\\
    \vdots & \cdots & \vdots\\
    \cpr{n}{1} & \cdots & \cpr{n}{n}\\
   \end{array}
\right]
\end{eqnarray}
\end{definition}

In this article, we normalize the PageRank centrality so that
  the sum of the centrality values over all nodes is equal to $n$.
Let $\11$ denote  the $n$-dimensional vector of all 1s.
Then, the {\em PageRank centrality} of $G$ is the solution to the
  following {\em Markov random-walk}
   equation \cite{PageBMW98,Haveliwala03}:
\begin{eqnarray}\label{PR}
\PR{{\WW,\alpha}} = \alpha \cdot \11 + (1-\alpha)
  \cdot  \WW^T\left(\DD_{\WW}^{out}\right)^{-1}\PR{{\WW,\alpha}}
\end{eqnarray}

Because $\11 = \sum_u \11_u$, we have: 
\begin{proposition}[PageRank Conforming]\label{prop:basic1}
For any $G = (V,E,\WW)$ and $\alpha > 0$:
\begin{eqnarray}
\PR{{\WW,\alpha}} = \sum_{u\in V} \pp_{u} = 
  \PPR_{{\WW,\alpha}}^T \cdot \11
\end{eqnarray}
\end{proposition}

Because Markov processes preserve 
   the probability mass of the starting vector, we also have:

\begin{proposition}[Markovian Conforming]\label{prop:basic2}
For any $G = (V,E,\WW)$ and $\alpha > 0$, $\PPR_{{\WW,\alpha}}$ 
 is non-negative and:
\begin{eqnarray}
  \PPR_{{\WW,\alpha}} \cdot \11 = \11
\end{eqnarray}
\end{proposition}

In summary, the PageRank matrix $\PPR_{{\WW,\alpha}}$
  is a special matrix  associated with network $G$ --- its row sum is the
  vector of all 1s and its column sum is the PageRank  centrality of $G$.

\subsection{PageRank Completion of Symmetric Networks}

PageRank centrality and personalized PageRank matrix 
  apply to both directed and undirected weighted graphs.
Both Proposition \ref{prop:basic1} and Proposition \ref{prop:basic2}
  also hold generally.
In this subsection, we will focus mainly on undirected weighted networks.
In such a case, let $\DD_{\WW}$ be the diagonal matrix associated
  with weighted degrees $\dd_{\WW} = \WW\cdot \11$ and
  let $\MM_{\WW} = \DD_{\WW}^{-1}\WW$ be 
  the standard random-walk transition matrix  on $G$.

To state the theorem below, let's first review a basic concept
  of Markov chain.
Recall that a {\em Markov chain} over $V$ is defined by
  an $n\times n$ transition matrix $\MM$ satisfying 
  the {\em stochastic condition}: $\MM$ is non-negative and
$\MM\cdot \11 = \11.$
A probability vector $\ppi$ is the {\em stationary distribution} of
  this Markov process if:
\begin{eqnarray}
\MM^T\ppi = \ppi
\end{eqnarray}
It is well known that 
  every irreducible and ergodic Markov chain has a stationary distribution.
Markov chain $\MM$ is {\em detailed-balanced}  if:
\begin{eqnarray}
\ppi[u] \MM[u,v] = \ppi[v] \MM[v,u], \quad \forall\ u,v\in V
\end{eqnarray}

We will now prove the following structural result:
\begin{theorem}[PageRank Completion]\label{theo:PageRankEssence}
For any weighted directed graph $G = (V,E,\WW)$  and restart constant~$\alpha > 0$:
\begin{itemize}
\item[{\bf A:}]
$\PPR_{{\WW,\alpha}}$ and $\left(\DD_{\WW}^{out}\right)^{-1}\cdot \WW$
  have the same eigenvectors. 
  Thus, both Markov chains have the same stationary distribution.
\item [{\bf B:}] $\PPR_{{\WW,\alpha}}$ is detailed-balanced
  if and only if $\WW$ is symmetric. 
\end{itemize}
Furthermore, when $\WW$ is symmetric, let 
  $\Gbar_{\alpha}= (V,\Ebar_{\alpha},\WWbar_{\alpha})$ be the
  affinity network such that:
\begin{eqnarray}
\WWbar_{\alpha} = \DD_{\WW} \cdot \PPR_{{\WW,\alpha}}  \quad \mbox{\rm and} \quad
\Ebar =\{(u,v): \WWbar_{\alpha}[u,v] > 0\}
\end{eqnarray}
Then, $\Gbar_{\alpha}$ satisfies the following conditions:
\begin{enumerate}
\item {\bf Symmetry Preserving:} 
  $\WWbar^T = \WWbar$, i.e., $\Gbar_{\alpha}$ is an undirected affinity network.
\item {\bf Complete Information}: If $G$ is connected, then
  $\Ebar_{\alpha}$ is a 
  complete graph with $|V|$ self-loops.
\item {\bf Degree and Stationary Preserving}:  
$\WW\cdot \11 = \WWbar\cdot \11$. Thus, $\DD_{\WW} = \DD_{\WWbar}$
  and the random-walk Markov chains $\MM_{\WW}$ and $\MM_{\WWbar}$ 
   have the same stationary distribution.
\item {\bf Markovian and PageRank Conforming}: 
\begin{eqnarray}
\MM_{\WWbar} \cdot \11 = \11
\quad \mbox{\rm and} \quad \MM_{\WWbar}^T\cdot \11 = \PR{{\WW,\alpha}}
\end{eqnarray}
\item {\bf Simultaneously Diagonalizable}:
For any symmetric $\WW$,
  recall  $\LLL_{\WW} = \DD_{\WW}-\WW$ denotes the Laplacian matrix associated
  with $\WW$.
Let $\LL_{\WW} = 
\DD_{\WW}^{-\frac{1}{2}} \LLL_{\WW} \DD_{\WW}^{\frac{1}{2}}  = 
  \II - \DD_{\WW}^{-\frac{1}{2}} \WW  \DD_{\WW}^{-\frac{1}{2}}$
be the normalized Laplacian matrix associated with $\WW$.
Then,  $\LL_{\WW}$ and $\LL_{\WWbar}$ are simultaneously diagonalizable.
\item {\bf Spectral Densification and Approximation}: 
For all $\xx \in \Reals{n}$:
\begin{eqnarray}
\alpha\cdot \LLL_{\WW}  \leq \xx^T \left(\frac{1}{1-\alpha}\cdot \LLL_{\WWbar} \right)\xx \leq \frac{1}{\alpha}\LLL_{\WW}\\
\alpha\cdot \LL_{\WW}  \leq \xx^T \left(\frac{1}{1-\alpha}\cdot\LL_{\WWbar}\right) \xx \leq \frac{1}{\alpha}\LL_{\WW}
\end{eqnarray}
In other words,
  $G$ and $\frac{1}{1-\alpha}\cdot\Gbar_{\alpha}$ are 
  $\frac{1}{\alpha}$-spectrally similar.
\end{enumerate}
\end{theorem}
\noindent{\bf Remarks}
{\em We rescale $\LLL_{\WWbar}$ and $\LL_{\WWbar}$
  by $\frac{1}{1-\alpha}$ because $\Gbar_{\alpha}$ has self-loops of 
  magnitude $\alpha \DD_{\WW}$. 
In other words,
 $\Gbar_{\alpha}$ only uses $(1-\alpha)$ fraction of its 
 weighted degrees for connecting different nodes in $V$.}

\begin{proof}
Let $n = |V|$.
For any initial distribution $\ss$ over $V$, we 
  can explicitly express $\pp_{\ss}$ as:
\begin{eqnarray}\label{eqn:pprSeries}
 \pp_{\ss} = \alpha \sum_{k=0}^{\infty} (1-\alpha)^k 
\cdot  \left(\WW^T\cdot\left(\DD_{\WW}^{out}\right)^{-1}\right)^k \cdot \ss
 \end{eqnarray}

Consequently: we can express $\PPR_{{\WW,\alpha}}$ as:
 \begin{eqnarray}\label{eqn:pprMatrixSeries}
\PPR_{{\WW,\alpha}}
 = \alpha \sum_{k=0}^{\infty} (1-\alpha)^k 
\cdot  \left(\left(\DD_{\WW}^{out}\right)^{-1}\cdot \WW\right)^k
 \end{eqnarray}
Note that $\alpha \sum_{k=0}^{\infty} (1-\alpha)^k  = 1$.
Thus, $\PPR_{{\WW,\alpha}}$ is 
  a convex combination of (multi-step) random-walk 
  matrices defined by $\left(\DD_{\WW}^{out}\right)^{-1}\cdot\WW$.
Statement A 
  follows directly from the fact that 
  $\left(\left(\DD_{\WW}^{out}\right)^{-1}\cdot\WW\right)^k$ 
  is a stochastic matrix for any integer  $k\geq 0$.

The following fact is well known \cite{aldous2002reversible}:
\begin{quote}
{\em 
Suppose $\MM$ is a Markov chain with stationary distribution $\ppi$.
Let $\M{\Pi}$ be the diagonal matrix defined by $\ppi$.
Then, $\MM^T\M{\Pi}$ is symmetric if and only if
  the Markov process defined by $\MM$ is  detailed balanced.
}
\end{quote}

We now assume $\WW = \WW^T$.
Then,  Eqn. (\ref{eqn:pprMatrixSeries}) becomes:
 \begin{eqnarray}\label{eqn:pprMatrixSeriesSym}
\PPR_{{\WW,\alpha}}
 = \alpha \sum_{k=0}^{\infty} (1-\alpha)^k 
\cdot  \left(\DD_{\WW}^{-1}\cdot \WW\right)^k
 \end{eqnarray}
The stationary distribution of $\DD_{\WW}^{-1}\WW$ --- and 
 hence of $\PPR_{{\WW,\alpha}}$ --- 
  is proportional to $\dd = \WW\cdot\11$.
$\PPR_{{\WW,\alpha}}$ is detailed balanced because 
  $\WWbar =  \DD_{\WW}\cdot \PPR_{{\WW,\alpha}}$ is a symmetric matrix.
Because $\left(\left(\DD_{\WW}^{out}\right)^{-1}\cdot \WW\right)^k$ (for all positive integers) have
  a common stationary distribution, $\PPR_{{\WW,\alpha}}$ is not detailed balanced when 
  $\WW$ is not symmetric.
It is also well known --- by Eqn. (\ref{eqn:pprSeries}) --- that
 for all $u,v\in V$, 
  $\PPR_{\WW,\alpha}[u,v]$ is equal to the probability that a run 
  of random walk starting at $u$ passes by $v$ immediately 
  before it restarts.
Thus, when $G$ is connected, 
  $\PPR_{{\WW,\alpha}}[u,v]>  0$ for all $u,v\in V$.
Thus, $\nnz(\WWbar_{\alpha}) = n^2$, and 
  $\Ebar_{\alpha}$, the nonzero pattern of  $\WWbar_{\alpha}$, is a 
  complete graph with $|V|$ self-loops.
We have now established Condition B and Conditions 1 - 4.

We now prove Conditions 5 and 6.\footnote{Thanks to Dehua Cheng 
 of USC for assisting this proof.}
Recall that when $\WW = \WW^T$, 
  we can express the personalized PageRank matrix as:
 \begin{eqnarray*}
\PPR_{{\WW,\alpha}}   = \alpha \sum_{k=0}^{\infty} (1-\alpha)^k 
\cdot  \left(\DD_{\WW}^{-1}\cdot\WW\right)^k.
 \end{eqnarray*}

Thus: 
\begin{eqnarray*}
\WWbar_{\alpha}   =  \DD_{\WW}\cdot \PPR_{{\WW,\alpha}}
  =  \left(\alpha \sum_{k=0}^{\infty} (1-\alpha)^k \cdot  
  \DD_{\WW}\cdot \left(\DD_{\WW}^{-1}\WW\right)^k\right).
\end{eqnarray*}
We compare the Laplacian matrices associated with $\WW$ and $\WWbar$:
\begin{eqnarray*}
\LLL_{\WW} & = &\DD_{\WW} - \WW = \DD_\WW^{1/2}
  \left(\II - \DD_{\WW}^{-1/2}\WW\DD_{\WW}^{-1/2}\right)\DD_{\WW}^{1/2}
 = \DD_\WW^{1/2}  \LL_{\WW}\DD_{\WW}^{1/2}.\\
\LLL_{\WWbar} & = &\DD_{\WW} - \WWbar 
 =  \DD_\WW^{1/2}  \LL_{\WWbar}\DD_{\WW}^{1/2}
\end{eqnarray*}
where 
\begin{eqnarray*}
\LL_{\WWbar} = \II - \alpha \sum_{k=0}^{\infty} (1-\alpha)^k \cdot 
(\DD_{\WW}^{-1/2}\WW\DD_{\WW}^{-1/2})^{k}.
\end{eqnarray*}
Let $\lambda_1 \geq \lambda_{2} \geq ... \geq \lambda_{n}$ be the $n$ eigenvalues
  of $\DD_{\WW}^{-1/2}\WW\DD_{\WW}^{-1/2}$.
Let $\uu_1, ..., \uu_n$  
  denote the unit-length eigenvectors of $\DD_{\WW}^{-1/2}\WW\DD_{\WW}^{-1/2}$ 
  associated with 
  eigenvalues $\lambda_1, \cdots, \lambda_n$, respectively.
We have $|\lambda_i|\leq 1$.
Let $\M{\Lambda}$ be the diagonal matrix associated with $(\lambda_1,...,\lambda_n)$ and $\M{U} = [\uu_1,...,\uu_n]$.
By the spectral theorem --- i.e., the eigenvalue decomposition for symmetric matrices ---
  we have:
\begin{eqnarray}
\UU^T\DD_{\WW}^{-1/2}\WW\DD_{\WW}^{-1/2}\UU & = & \M{\Lambda} \\
\UU \UU^T  =  \UU^T \UU  & = & \II
\end{eqnarray}
Therefore:
\begin{eqnarray*}
\LLL_{\WW} & = &
 \DD_\WW^{1/2}\UU\UU^T
  \left(\II - \DD_{\WW}^{-1/2}\WW\DD_{\WW}^{-1/2}\right)\UU\UU^T\DD_{\WW}^{1/2}\\
& = & 
 \DD_\WW^{1/2}\UU
  \left(\II - \UU^T\DD_{\WW}^{-1/2}\WW\DD_{\WW}^{-1/2}\UU\right)\UU^T\DD_{\WW}^{1/2}\\
 & =& \DD_\WW^{1/2}\UU
  \left(\II - \M{\Lambda}\right)\UU^T\DD_{\WW}^{1/2}.
\end{eqnarray*}
Similarly:
\begin{eqnarray*}
\LLL_{\WWbar_{\alpha}} & = & \DD_{\WW} - \WWbar_{\alpha} 
= \DD_\WW^{1/2}  \LL_{\WWbar}\DD_{\WW}^{1/2}\\
& = & 
\DD_\WW^{1/2}
\left(
\II - \alpha \sum_{k=0}^{\infty} (1-\alpha)^k \cdot 
 (\DD_{\WW}^{-1/2}\WW\DD_{\WW}^{-1/2})^{k}\right)
\DD_{\WW}^{1/2}\\
& = & 
\DD_\WW^{1/2}\UU 
\left(
\II - \alpha \sum_{k=0}^{\infty} (1-\alpha)^k \cdot 
\UU^T (\DD_{\WW}^{-1/2}\WW\DD_{\WW}^{-1/2})^{k}\UU\right)
\UU^T\DD_{\WW}^{1/2}\\
& = & \DD_\WW^{1/2}\UU 
\left(
\II - \alpha \sum_{k=0}^{\infty} (1-\alpha)^k \cdot \M{\Lambda}^k
\right)
\UU^T\DD_{\WW}^{1/2}\\
& = & \DD_\WW^{1/2}\UU 
\left(
\II - \frac{\alpha}{\II - (1-\alpha)\M{\Lambda}}
\right)
\UU^T\DD_{\WW}^{1/2}.
\end{eqnarray*}
The derivation above has proved  Condition (5).
To prove Condition (6), consider an arbitrary
  $\xx \in \Reals{n}\setminus\{\00\}$.
With $\yy = \UU^T\DD_{\WW}^{1/2}\xx$, we have:
\begin{eqnarray*}
\frac{\xx^T\frac{1}{1-\alpha}\LLL_{\WWbar}\xx}{\xx^T\LLL_{\WW}\xx}
& = & \frac{1}{1-\alpha}\cdot \frac{\xx^T\DD_{\WW}^{1/2}\UU
\left(\II - \frac{\alpha}{\II - (1-\alpha)\M{\Lambda}}\right)
 \UU^T\DD_{\WW}^{1/2}\xx}
{\xx^T\DD_{\WW}^{1/2}\UU\left(\II - \M{\Lambda}\right)\UU^T\DD_{\WW}^{1/2}\xx}\\
& = & \frac{1}{1-\alpha}\cdot \frac{\yy^T
\left(\II - \frac{\alpha}{\II - (1-\alpha)\M{\Lambda}}\right)
\yy}{\yy^T\left(\II - \M{\Lambda}\right)\yy}
\end{eqnarray*}
This ratio is in the interval of:
\[
\left[\inf_{\lambda: |\lambda|\leq 1} \frac{1}{1-(1-\alpha)\lambda},
\sup_{\lambda: |\lambda|\leq 1} \frac{1}{1-(1-\alpha)\lambda}\right]
= \left[\frac{1}{2-\alpha},\frac{1}{\alpha}\right].
\]
\end{proof}

\subsection{PageRank Completion, Community Identification, and Clustering}

PageRank completion has an immediate application to 
  the community-identification approaches developed in \cite{BCMTITCS,B3CT}.
This family of methods first constructs a preference network from
  an input weighted graph $G = (V,E,\WW)$.
It then applies various social-choice aggregation functions 
  \cite{ArrowBook} to define
  network communities \cite{BCMTITCS,B3CT}.
In fact, Balcan {\em et al} \cite{B3CT} show that 
 the PageRank completion of $G$  provides a wonderful scheme
  (see also  in Definition \ref{Interpretation})
  for constructing preference networks from affinity networks.

In addition to its classical connection with PageRank centrality, 
  PageRank completion also has 
  a direct connection with network clustering.
To illustrate this connection, let's recall a well-known  approach in spectral graph theory for clustering
\cite{Cheeger,SpielmanTengSpectral,LovaszSimonovits,SpielmanTengLocalClustering,AndersenChungLang}:

\begin{algorithm}[H]
\caption*{{\bf Algorithm}: {\sf Sweep}$(G, \vv)$}
\nonumber
\begin{algorithmic}[1]
\REQUIRE $G = (V,E,\WW$)  and  $\vv \in \Reals{|V|}$

\STATE Let $\ppi$ be an ordering of $V$ according to $\vv$, i.e.,
 $\forall k\in [n-1]$,  $\vv[\pi(k)] \geq \vv[\pi(k+1)]$
\STATE Let $S_k = \{\pi(1),...,\pi(k)\}$

\STATE
Let $k^* = {\rm arg min}_k\ \conductance_{\WW}(S_k)$. 
\STATE Return $S_{k^*}$
\end{algorithmic}
\end{algorithm}

Both in theory and in practice, 
  the most popular vectors used in {\sf Sweep} are:
\begin{itemize}
\item {\bf Fiedler vector:} the eigenvector associated with the
  second smallest eigenvalue of the Laplacian matrix $\LLL_{\WW}$ \cite{FiedlerAlg,Fiedler,SpielmanTengSpectral}.
\item {\bf Cheeger vector:} $\DD_\WW^{-1/2}\V{v}_2$, where $\V{v}_2$ is the eigenvector associated with the
  second smallest eigenvalue of the normalized Laplacian matrix $\LL_{\WW}$ 
  \cite{Cheeger,ChungBook}.
\end{itemize}

The sweep-based clustering method and Fiedler/Cheeger vectors
  are the main subject of following
   beautiful theorem \cite{Cheeger} in spectral graph theory:

\begin{theorem}[Cheeger's Inequality]
\label{theo:Cheeger}
For any symmetric weighted graph $G = (V,E,\WW)$,  
  let $\lambda_{2}$ be the second smallest eigenvalue of
 the normalized Laplacian matrix $\LL_{\WW}$ of $G$.
Let $\V{v}_2$ be the eigenvector associated with $\lambda_2$
   and $S = {\sf Sweep}(G, \DD_\WW^{-1/2} \vv_2)$.
Then: 
\begin{eqnarray}
\frac{\lambda_2}{2} \leq \conductance_{\WW}(S) \leq \sqrt{2\lambda_{2}}
\end{eqnarray}
\end{theorem}

By Theorem \ref{theo:PageRankEssence}, 
 the normalized Laplacian matrices of $G$ and its PageRank completion
  are simultaneously diagonalizable.
Thus, we can also use the eigenvector of the PageRank completion of $G$
  to identify a cluster of $G$ whose conductance is guaranteed by the 
  Cheeger's inequality.

{\em \begin{quote}                                                      
Then, how is the PageRank completion necessarily a better representation
  of the information contained in the original network?        

For example, with respect to network clustering, what desirable properties
  does the PageRank completion have that the original graph doesn't?
\end{quote}}

While we are still looking for a comprehensive answer to these questions,
  we will now use the elegant result of 
  Andersen, Chung, and Lang \cite{AndersenChungLang}
  to illustrate that the PageRank completion indeed contains more {\em direct}
  information about network clustering than the original data $\WW$.
Andersen {\em et al} proved that if one applies sweep 
to vectors $\{\DD^{-1}_{\WW}\cdot \pp_v\}_{v\in V}$, 
  then one can obtain a cluster whose conductance is nearly as small as
  that guaranteed by Cheeger's inequality.
Such a statement does not hold for the rows in
   the original network data $\WW$, particularly 
  when $\WW$ is sparse.

In fact, the result of 
  Andersen, Chung, and Lang \cite{AndersenChungLang}
  is much stronger.
They showed that for any cluster $S \subset V$, 
  if one selects a random node $v\in S$ with probability
   proportional to the weighted degree $d_v$ of the node, 
   then, with probability at least $1/2$, 
   one can identify a cluster $S'$ of conductance at most 
  $O(\sqrt{\conductance_{\WW}(S)\log n})$ by
  applying sweep to vector $\DD^{-1}_{\WW}\cdot \pp_v$.
In other words, the row vectors in the PageRank completion --- i.e., 
  the personalized PageRank vectors that 
  represent the individual data associated with nodes ---
  have rich and direct information about network clustering 
  (measured by conductance).
This is a property that the original network data simply doesn't  have,
  as one is usually not able to  identify good clusters
  directly from the individual rows of $\WW$.

In summary, Cheeger’s inequality 
  and its algorithmic proof can be viewed as the mathematical foundation 
  for {\em global} spectral partitioning, 
  because the Fiedler/Cheeger vectors are 
  formulated from the network data as a whole.
From this global perspective, both 
  the original network and its PageRank completion are equally effective.
In contrast, from the local perspective of individual-row data,
  Andersen, Chung, and Lang's result highlights the 
  effectiveness of the PageRank completion to
  {\em local clustering} \cite{SpielmanTengLocalClustering}:
The row data associated with nodes in the PageRank completion 
  provides effective information for identifying good clusters.
Similarly, from the corresponding column in the PageRank completion,
  one can also directly and ``locally'' obtains
  each node's PageRank centrality.
In other words, PageRank completion transforms the input network data $\WW$
  into a ``complete-information'' network model $\WWbar$, and 
  in the process, it distilled the centrality/clusterability information
   implicitly embedded globally in $\WW$
   into an ensemble of nodes' ``individual'' network data
   that explicitly encodes the centrality information and 
   locally capturing the clustering structures.

\section{Connecting Multifaceted Network Data}\label{sec:RW}

The formulations highlighted in 
  Section \ref{sec:RichnessOther},
   such as the cooperative, incentive, powerset, and preference models,
  are just a few examples of network models beyond the traditional
  graph-based framework.
Other extensions include the 
  popular probabilistic graphical model  \cite{KollerFriedman} and
  game-theoretical graphical model \cite{KearnsGraphicalGame,DGP,ChenDengTeng}.
These models use relatively homogeneous node and edge types,
  but nevertheless represent a great source of expressions
  for multifaceted and multimodal network data.

While diverse network models enable us to express multifaceted network data,
  we need mathematical and algorithmic tools to connect them.
For some applications such as community identification, one may need to properly 
  use some data facets as metadata to evaluate or cross validate 
  the network solution(s) identified from the main network facets \cite{MetaData}.

\begin{quote}
{\em
But more broadly, for many real-world network analysis tasks, we need 
  a systematic approach to {\em network composition}
 whose task is to integrate the multifaceted data into a single effective
   network worldview.
Towards this goal, a basic theoretical step in multifaceted network analysis is 
  to establish a unified worldview for capturing 
  multifaceted network data expressed in various models.}
\end{quote}

Although fundamental, formulating a unified worldview 
  of network models is still largely an outstanding research problem.
In this section, 
  we sketch our preliminary studies in using Markov chains
  to build a “common platform” for the network models discussed 
  in Section \ref{sec:RichnessOther}. 
We hope this study will inspire a general theory for data 
  integration, network composition, and multifaceted network analysis.
We also hope that it will
   help to strengthen the connection between various fields,
  as diverse as statistical modeling, geometric embedding,
  social influence, network dynamics, game theory, and 
  social choice theory,
  as well as various application domains
  (protein-protein interaction, viral marketing,
  information propagation, electoral behavior,
  homeland security, healthcare, {\em etc.}),
  that have provided different but valuable techniques and
  motivations to network analysis.

\subsection{Centrality-Conforming Stochastic Matrices of Various Network Models}
\label{sec:CentralityConfirming}

Markov chain --- a basic statistical model --- is also 
 a fundamental network concept.
For a weighted network $G = (V,E,\WW)$,
  the standard random-walk transition  
 $\left(\DD_{\WW}^{out}\right)^{-1}\cdot \WW$ 
  is the most widely-used stochastic matrix associated with $G$.
Importantly, Section \ref{Sec:NE} illustrates that
 other Markov chains
 --- such as PageRank Markov chain $\PPR_{{\WW,\alpha}}$ ---
  are also natural with respect to network data $\WW$.
Traditionally, a Markov chain is characterized by
  its stochastic condition, stationary distribution, mixing time, and
  detailed-balancedness.
Theorem \ref{theo:PageRankEssence} highlights another
  important feature of Markov chains in the context of 
  network analysis: 
{\em The PageRank Markov chain is conforming with respect to 
   PageRank centrality,} that is, 
  for any network $G = (V,E,\WW)$ and $\alpha > 0$, we have:
$$\PPR_{{\WW,\alpha}}^T\cdot \11 = \PR{{\WW,\alpha}}.$$

\begin{quote}
{\em How should we derive stochastic matrices from other network models?
Can we construct Markov chains that are
  centrality-confirming with respect to
  natural centrality measures of these network models?
}
\end{quote}

In this section, 
  we will examine some centrality-confirming Markov chains that can be
  derived from network data given by 
  preference/incentive/cooperative/powerset  models.
\subsubsection*{{\sc The Preference Model}}

For the preference model, there is a family of natural Markov
  chains, based on weighted 
  aggregations in social-choice theory \cite{ArrowBook}.
For a fixed $n$, let $\ww  \in (\mathbb{R^+}\cup \{0\})^n$
  be a non-negative and monotonically non-increasing vector.
For the discussion below, we will assume that $\ww$ is normalized 
  such that $\sum_{i=1}^n \ww[i] = 1$.
For example, while the famous Borda count \cite{YoungBorda} 
  uses $\ww = [n,n-1,...,1]^T$,
  the normalized Borda count uses $\ww = [n,n-1,...,1]^T/{n \choose 2}$.

\begin{proposition}[Weighted Preference Markov Chain]
Suppose $A = (V,\Pi)$ is a {\em preference network} 
  over $V = [n]$ and
  $\ww$ is non-negative and monotonically non-increasing weight vector,
  with $||\ww||_1 = 1$.
Let $\MM_{A,\ww}$ be the matrix  in which for each $u\in V$, the
  $u^{th}$ row of $\MM_{A,\ww}$ is: 
$$\V{\pi}_u\circ \ww = [\ww[\V{\pi}_u(1)],...,\ww(\V{\pi}_u(n))].$$
Then, $\MM_{A,\ww}$  defines a Markov chain, i.e.,
   $\MM_{A,\ww}\11 = \11$.
\end{proposition}
\begin{proof}
$\MM_{A,\ww}$ is a stochastic matrix
  because each row of $\MM_{A,\ww}$ is a permutation of $\ww$, 
  and permutations preserve the L1-norm of the vector.
\end{proof}

Social-choice aggregation based on $\ww$ also defines
  the following natural centrality measure, 
  which can be viewed as the collective ranking over $V$ based on the
  preference profiles of $A = (V,\Pi)$: 
\begin{eqnarray}
\centrality_{\Pi,\ww}[v] = \sum_{u\in V} \ww[\pi_u(v)]
\end{eqnarray}

Like PageRank Markov chains, 
  weighted preference Markov chains
  also enjoy the centrality-conforming property:
\begin{proposition}
For any preference network $A=(V,\Pi)$, in which $\Pi \in \LV^{|V|}$:
\begin{eqnarray}
\MM_{A,\ww}^T \cdot \11 = \centrality_{\Pi,\ww}
\end{eqnarray}
\end{proposition}

\subsubsection*{{\sc The Incentive Model}}

We now focus on a special family of incentive networks:
We assume for  $U = (V,\bvec{u})$ and $s\in V$:
\begin{enumerate}
\item   $u_s$ is monotonically non-decreasing, i.e., for all $T_1\subset T_2$,
  $u_s(T_1) \leq u_s(T_2)$.
\item  $u_s$ is normalized, i.e., $u_s(V\setminus \{s\}) = 1$.
\end{enumerate}

Each incentive network defines a natural cooperative network, 
  $H_{U} = (V,\V{\tau}_{SocialUtility})$:
For  any $S\subseteq V$, let the {\em social utility}
  of $S$ be:
\begin{eqnarray}
\V{\tau}_{SocialUtility}(S)  =  \sum_{s\in S} u_{s}(S\setminus \{s\})
\end{eqnarray}

The Shapley value \cite{ShapleyValue} --- a classical game-theoretical 
  concept --- 
  provides a natural centrality measure for cooperative networks.
\begin{definition}[Shapley Value]
Suppose $\V{\tau}$ is the characteristic 
  function of a cooperative game over $V =[n]$.
Recall that $\LV$  denotes the set of all permutations of $V$.
Let $S_{\V{\pi}, v}$ denotes the set of players preceding $v$ in a permutation $\V{\pi} \in \LV$.
Then, the {\em Shapley value} $\V{\phi}_{\V{\tau}}^{Shapley}[v]$  of a player $v\in V$ is:
\begin{eqnarray}\label{eqn:Shapley}
\V{\phi}_{\V{\tau}}^{Shapley}[v] = \expec{\V{\V{\pi}}\sim \LV}{\V{\tau}[S_{\V{\pi},v} \cup \{v\}] - 
  \V{\tau}[S_{\V{\pi},v}]}
\end{eqnarray}
\end{definition}

The Shapley value $\V{\phi}_{\V{\tau}}^{Shapley}[v]$  
  of player $v\in V$  is
   the  {\em expected marginal contribution} of $v$ 
   over the set preceding $v$  in a random permutation of the players.    
The Shapley value has many attractive properties, and 
 is widely considered to be the {\em fairest} measure of 
  a player's power index in a cooperative~game.

We can use Shapley values to define both the stochastic matrix and 
  the centrality of incentive networks $U$.
Let $\centrality_{U}$ be the Shapley value of the cooperative game
  defined by $\V{\tau}_{SocialUtility}$.
Note that the incentive network $U$
  also defines $|V|$ natural individual cooperative networks:
For each $s\in V$ and $T\subset V$, let:
\begin{eqnarray}
\V{\tau}_{s}(T)  =  
\left\{\begin{array}{ll}
u_s(T\setminus \{s\}) & \mbox{if $s\in T$}\\
0 & \mbox{if $s\not\in T$}
\end{array}
\right.
\end{eqnarray}

\begin{proposition}[The Markov Chain of Monotonic Incentive Model]
\label{prop:MCIC}
Suppose $U = (V,\bvec{u})$ is an incentive network over $V = [n]$,
 such that $\forall s\in V$,
 $u_s$ is monotonically non-decreasing and $u_s(V\setminus \{s\}) = 1$.
Let $\MM_{U}$ be the matrix in which for each $s\in V$, the
  $s^{th}$ row of $\MM_{U}$ is the Shapley value of
  the cooperative game with characteristic function $\V{\tau}_{s}$.
Then, $\MM_U$ defines a Markov chain and is centrality-conforming 
  with respect to $\centrality_{U}$, i.e., 
  (1) $\MM_U\11 = \11$ and (2) $\MM_U^T\11 = \centrality_{U}$.
\end{proposition}
\begin{proof}
This proposition is the direct consequence of two basic 
  properties of Shapley's beautiful characterization \cite{ShapleyValue}:
\begin{enumerate}
\item The Shapley value  is {\em efficient}:  
    $\sum_{v\in V} \phi_{\V{\tau}}[v] = \V{\tau}(V)$.
\item The Shapley value is  {\em Linear}: 
For any two characteristic functions $\V{\tau}$ and $\V{\omega}$, 
$\phi_{\V{\tau} + \V{\omega}} = \phi_{\V{\tau}} + \phi_{\V{\omega}}$.
\end{enumerate}
By the assumption  $u_s$ is monotonically non-decreasing,
  we can show that every entry of the Shapley value
  (as given by Eqn: (\ref{eqn:Shapley})) is non-negative.
Then, it follows from
   the efficiency of Shapley values and the assumption that
  $\forall s\in V, u_s(V\setminus \{s\}) = 1$,
  that $\MM_{U}$ is a stochastic matrix, and hence it
   defines a Markov chain.
Furthermore, we have:
\begin{eqnarray}
\V{\tau}_{SocialUtility} = \sum_{s\in V} \V{\tau}_s
\end{eqnarray}
Because $\centrality_{U}$ is the Shapley value of the 
  cooperative game with characteristic function 
$\V{\tau}_{SocialUtility}$,
  the {\em linearity} of the Shapley value then implies
  $\MM_U^T\11 = \centrality_{U}$, i.e., 
  $\MM_U$ is centrality-conforming with respect to $\centrality_{U}$.
\end{proof}

\subsubsection*{{\sc The Influence Model}}

Centrality-conforming Markov chain can also be
  naturally constructed for a family of powerset networks.
Recall from  Section \ref{sec:RichnessOther} that
  an influence process $\calD$ and social network $G=(V,E)$
  together define a powerset network, 
 $\V{P}_{G,\calD}: 2^V\times 2^V\rightarrow [0,1]$,
where for each $T \in 2^V$, 
 $\V{P}_{G,\calD}[S,T]$ 
 specifies the probability that $T$ is the final activated set 
  when $S$ cascades its influence through $G$.
As observed in \cite{ChenTeng}, the influence model also
  defines a natural cooperative game, whose
  characteristic function is the 
  influence spread function:
$$\V{\sigma}_{G,\calD}(S) = \sum_{T\subseteq V} |T|\cdot \V{P}_{G,\calD}[S,T], 
\quad \forall S\subseteq V.$$
Chen and Teng \cite{ChenTeng} proposed to use the Shapley value 
  of this social-influence game as a centrality measure
  of the powerset network defined by  $\V{P}_{G,\calD}$.
They showed that this social-influence centrality measure, to be denoted
  by $\centrality_{G,\calD}$,
  can be uniquely characterized by a set of five natrual axioms \cite{ChenTeng}.
Motivated by the PageRank Markov chain, they also
  constructed the following centrality-conforming Markov chain
  for social-influence models.

\begin{proposition}[Social-Influence Markov Chain]
Suppose $G = (V,E)$ is a social network and $\calD$
  is a social-influence process.
Let $\MM_{G,\calD}$ be the matrix in which for each $v\in V$,
  the $v^{th}$ row of $\MM_{G,\calD}$
  is given by the Shapley value
  of the cooperative game with the following 
  characteristic function:
\begin{eqnarray}
\V{\sigma}_{G,\calD,v}(S) = \sum_{T\subseteq V} \V{[} v\in T\V{]}\cdot 
  \V{P}_{G,\calD}[S,T]
\end{eqnarray}
where $\V{[} v\in T\V{]}$ is the indicator function for event ($v\in T$).
Then, $\MM_{G,\calD}$ defines a Markov chain and is 
  centrality-conforming with respect to $\centrality_{G,\calD}$, 
  i.e., (1) $\MM_{G,\calD}\11 = \11$ and (2) 
  $\MM_{G,\calD}^T\11 = \centrality_{G,\calD}$.
\end{proposition}
\begin{proof}
For all $v\in V$, the characteristic function $\V{\sigma}_{G,\calD,v}$
  satisfies the following two conditions:
\begin{enumerate}
\item $\V{\sigma}_{G,\calD,v}$ is monotonically non-decreasing.
\item $\V{\sigma}_{G,\calD,v}(V) = 1$.
\end{enumerate}
The rest of the proof is essentially the same as the 
   proof of Proposition \ref{prop:MCIC}.
\end{proof}

\subsection{Networks Associated with Markov Chains}

The common feature in the Markovian formulations 
  of Section \ref{sec:CentralityConfirming} 
  suggests the possibility of a general
  theory that various network models beyond graphs can be 
  succinctly analyzed through the worldview of Markov chains.
Such analyses are forms of dimension reduction of network data ---
  the Markov chains
   derived, such as from social-influence instances, usually
   have lower dimensionality than the original network models.
In dimension reduction of data, inevitably some information is lost.
Thus, which Markov chain is formulated from a particular network model
  may largely depend on through which mathematical lens we are looking
  at the network data.
The Markovian formulations of Section \ref{sec:CentralityConfirming} are largely 
  based on centrality formulations.
Developing a more general Markovian formulation theory 
  of various network models remains the subject of future research.

But once we can reduce the network models specifying various aspects
  of network data 
  to a collection of Markov chains representing the corresponding 
  network facets, 
   we effectively reduce multifaceted network analysis 
  to a potentially simpler task --- the analysis of multilayer networks
  \cite{PorterMultilayeredNetwork,LermanTengYan}.
Thus, we can apply various emerging techniques for multilayer network analysis
\cite{HarvardMultilayer,YaleMultilayer,PaulChenMultilayer} 
  and network composition \cite{LermanTengYan}.
We can further use standard techniques
  to convert the  Markov chains into weighted graphs
  to examine these network 
  models through the popular graph-theoretical worldview.

\subsubsection*{{\sc Random-Walk Connection}}

Because of the following characterization, the random-walk is 
  traditionally the most commonly-used connection
  between Markov chains and weighted networks.

\begin{proposition}[Markov Chains and Networks: Random-Walk Connection]
For any directed network $G = (V,E,\WW)$ in which every node
  has at least one out-neighbor, there is a unique 
  transition matrix:
$$\MM_{\WW} = \left(\DD_{\WW}^{out}\right)^{-1}\WW$$
   that captures the (unbiased) random-walk Markov process on $G$.
Conversely, given a transition matrix $\MM$, 
  there is an infinite family of weighted networks 
  whose random-walk Markov chains are consistent with $\MM$.
This family is given by:
$$ \{\GGamma\MM: \mbox{$\GGamma$ is a positive diagonal matrix}\}.$$
\end{proposition}

The most commonly-used diagonal scaling is $\PiPi$, the diagonal matrix
  of the stationary distribution.
This scaling is partially justified by the 
  fact that $\PiPi\MM$ is an undirected network
  if and only if  $\MM$ is a detailed-balanced Markov chain.
In fact in such a case, $\GGamma\MM$ is symmetric 
  if and only if there exists $c> 0$, $\GGamma = c\cdot \PiPi$.
Let's call $\PiPi\MM$ the {\em canonical Markovian network} of  
  transition matrix  $\MM$.
For a general Markov chain, we have:
\begin{eqnarray}
\11 \PiPi\MM = \V{\pi}^T \mbox{ and } \PiPi\MM\11 = \V{\pi}
\end{eqnarray}
Thus, although canonical Markovian networks are usually directed,
  their nodes always have the same in-degree and out-degree.
Such graphs are also known as the weighted Eulerian graphs. 

\subsubsection*{{\sc PageRank Connection}}

Recall that Theorem \ref{theo:PageRankEssence} features
  the derivation of 
   PageRank-conforming Markov chains from weighted networks.
In fact, Theorem \ref{theo:PageRankEssence} and its PageRank power series can be 
  naturally extended to any transition matrix $\MM$:
For any finite irreducible and ergodic Markov chain 
   $\MM$ and restart constant $\alpha > 0$,
  the matrix $\alpha \sum_{k=0}^{\infty} (1-\alpha)^k \cdot  \MM^k$
  is a stochastic matrix that preserves the detailed-balancedness,
   the stationary distribution, and the spectra of $\MM$.

Let's call  $\alpha \sum_{k=0}^{\infty} (1-\alpha)^k \cdot \PiPi \MM^k$
  the {\em canonical PageRank-Markovian network} of transition matrix $\MM$.
\begin{proposition}
For any Markov chain $\MM$, 
  the random-walk Markov chain of the 
  canonical PageRank-Markovian network
  $\alpha \sum_{k=0}^{\infty} (1-\alpha)^k \cdot \PiPi \MM^k$
   is conforming with respect to the PageRank of
   the canonical Markovian network $\PiPi\MM$.
\end{proposition}

\subsubsection*{{\sc Symmetrization}}

Algorithmically, computational/optimization problems on 
  directed graphs are usually harder than they are on undirected graphs.
For example, many recent breakthroughs in scalable graph-algorithm design
  are for limited to undirected graphs \cite{SpielmanTengLinear,SpielmanTengSpectralSparsification,SpielmanTengLocalClustering,AndersenChungLang,KoutisMillerPengII,KelnerLorenzoSidfordZhu,KelnerFlow,Sherman,PengFlow,USC}.
To express Markov chains as undirect networks, we can 
  apply the following well-known
  Markavian symmetrization formulation.
Recall a matrix $\M{L}$ is a {\em Laplacian matrix}
  if (1) $\M{L}$ is a symmetric matrix with non-positive off-diagonal entries,
  and (2) $\M{L}\cdot\11 = \00$.

\begin{proposition}[Canonical Markovian Symmetrization]\label{lem:MarkovianSymmetrization}
For any irreducible and ergodic finite Markov chain $\MM$:
\begin{eqnarray}
\PiPi - \frac{\PiPi\MM + \MM^T\PiPi}{2}
\end{eqnarray}
is a Laplacian matrix, where $\PiPi$ the diagonal matrix
  associated with $\MM$'s stationary distribution.
Therefore, $\frac{\PiPi\MM + \MM^T\PiPi}{2}$ is a
  symmetric network, whose degrees are normalized to stationary distribution
$\V{\pi} = \PiPi\cdot \11$.
When $\MM$ is detailed balanced, $\frac{\PiPi\MM + \MM^T\PiPi}{2}$ is 
  the {\em canonical Markovian network} of $\MM$.
\end{proposition}
\begin{proof}
We include a proof here for completeness.
Let $\ppi$ be the stationary distribution of $\MM$.
Then:
\begin{eqnarray*}
  \MM^T \ppi & = &  \ppi\\
  \PiPi \cdot \11 & = & \ppi  \\
  \MM\cdot \11 & = & \11
\end{eqnarray*}
Therefore:
\begin{eqnarray}
\left(\PiPi - \frac{\PiPi\MM + \MM^T\PiPi}{2}\right)\cdot \11
= \left(\ppi - \frac{\PiPi\11 + \MM^T\ppi}{2}\right)
= \00
\end{eqnarray}
The Lemma then follows from the fact that
$\frac{1}{2}(\PiPi\MM + \MM^T\PiPi)$ is symmetric and non-negative.
\end{proof}

Through the PageRank connection,
  Markov chains also have two extended Markovian symmetrizations:
\begin{proposition}[PageRank Markovian Symmetrization]\label{lem:MarkovianSymmetrization}
For any irreducible and ergodic finite Markov chain $\MM$ 
  and restart constant $\alpha > 0$, the two
matrices below:
\begin{eqnarray}
\PiPi - \alpha \sum_{k=0}^{\infty} (1-\alpha)^k \cdot\frac{ 
\PiPi\MM^k + (\MM^T)^k\PiPi}{2} \\
\PiPi - 
\alpha \sum_{k=0}^{\infty} (1-\alpha)^k \PiPi\cdot 
  \left(\PiPi^{-1} \cdot \frac{\PiPi\MM + \MM^T\PiPi}{2}\right)^k
\end{eqnarray}
are both Laplacian matrices. 
Moreover, 
  the second Laplacian matrix is $\frac{1}{\alpha}$-spectrally similar 
  to $(1-\alpha)\cdot \left(\PiPi - \frac{\PiPi\MM + \MM^T\PiPi}{2}\right)$.
\end{proposition}



\subsubsection*{{\sc Network Interpretations}} 

We now return to Balcan {\em et al.}'s approach \cite{B3CT}
  for deriving preference networks from affinity networks.
Consider the following natural extension of linear orders
   to express rankings with ties:
An {\em ordered partition} of $V$ is a total order of a partition 
  of $V$.
Let $\blinear{V}$ denote the set of all  {\em ordered partitions} of $V$:
For a $\sigma \in \bLV$, for $i,j\in V$, we 
$i$ is {\em ranked strictly ahead of} $j$ if $i$ and $j$ belong to
different partitions, and the partition containing $i$ is ahead of the
partition containing $j$ in $\sigma$.
If $i$ and $j$ are members of the same partition in $\sigma$, we say
  $\sigma$ is {\em indifferent of} $i$ and $j$.

\begin{definition}[PageRank Preferences] \label{Interpretation}
Suppose $G = (V,E,\WW)$ is a weighted graph and 
  $\alpha >0$ is a restart constant.
For each $u \in V$, let $\V{\pi}_u$ be the ordered partition 
  according to the descending ranking of $V$ 
  based on the personalized PageRank 
  vector $\pp_u = \PPR_{{\WW,\alpha}}[u,:]$.
We call $\Pi_{\WW,\alpha} =\{\V{\pi}_u\}_{u\in V}$  
  the {\em PageRank preference  profile of $V$} with respect to $G$,
and $A_{\WW,\alpha} = (V,\Pi_{\WW,\alpha})$ the 
{\em PageRank preference network} of $G$.
\end{definition}

As pointed out in \cite{B3CT}, other methods for deriving preference networks
  from weighted networks exist.
For example, one can obtain individual preference rankings
  by ordering nodes according to shortest path distances, 
  effective resistances, or maximum-flow/minimum-cut values.

\begin{quote}
 {\em Is the PageRank preference a
  desirable personalized-preference profile
  of an affinity network?
 }
\end{quote}

This is a basic question in network analysis.
In fact, much work has been done.
I will refer readers to the beautiful axiomatic approach of 
Altman and Tennenholtz 
  for characterizing 
  personalized ranking systems \cite{AltmanTennenholtzPersonalized}.
Although they mostly studied unweighted networks, many of their results can be
  extended to weighted networks.
Below, I will use Theorem \ref{theo:PageRankEssence}
  to address the following question that I was asked when 
  first giving a talk about PageRank preferences.
\begin{quote}
{\em 
By taking the ranking information from PageRank matrices ---
  which is usually asymmetric --- one may lose valuable network information.
For example, when $G = (V,E,\WW)$ is a undirected network,
 isn't it desirable to define ranking information according
  to  a symmetric matrix?
}
\end{quote}

At the time, I was not prepared to answer this question and replied
  that it was an excellent point.
Theorem \ref{theo:PageRankEssence} now provides an answer.
Markov chain theory uses an elegant concept to characterize
  whether or not a Markov chain $\MM$ has an undirected network realization.
Although Markov-chain transition matrices are usually asymmetric,
  if a Markov chain is detailed-balanced,
  then its transition matrix $\MM$ can be 
  diagonally scaled into a symmetric matrix by its stationary distribution.
Moreover, $\M{\Pi}\MM$ is the ``unique'' underlying undirected network
  associated with $\MM$.
By Theorem \ref{theo:PageRankEssence}, 
  $\PPR_{{\WW,\alpha}}$ is a Markov transition matrix with 
  stationary distribution $\DD_{\WW}$, and thus,
  $\WWbar_{\alpha} = \DD_{\WW}\cdot \PPR_{{\WW,\alpha}}$
  is symmetric if and only if $\WW$ is symmetric.
Therefore, because the ranking given by $\pp_u$ is the same as the ranking 
  given by $\WWbar[u,:]$, the PageRank preference profile
  is indeed derived from a symmetric matrix when $\WW$ is symmetric.

We can also define clusterability and other network models based on
  personalized PageRank matrices.
For example:
\begin{itemize}
\item {\bf PageRank conductance}:
\begin{eqnarray} \label{eqn:Conductance2}
\pagerankconductance_{\WW}(S) :=
\frac{\sum_{u\in S,v\not\in S} \WWbar[u,v]}
{\min\left(\sum_{u\in S,v\in V} \WWbar[u,v], \sum_{u\not\in S,v\in V} \WWbar[u,v]\right)}
\end{eqnarray}

\item {\bf PageRank utility}:
\begin{eqnarray} 
\pagerankutility_{\WW}(S) := \sum_{u\in S,v\in S} \PPR_{{\WW,\alpha}}[u,v]
\end{eqnarray}

\item {\bf PageRank clusterability}:
\begin{eqnarray} 
\pageranclusterability_{\WW}(S) := \frac{\pagerankutility_{\WW}(S)}{|S|}
\end{eqnarray}
\end{itemize}
Each of these functions defines a cooperative network
    based on $G = (V,E,\WW)$.
These formulations are connected with the PageRank of $G$. 
For example,  the Shapley value of the cooperative network given by 
  $\V{\tau} = \pagerankutility_{\WW}$ is the PageRank of $G$.

$\PPR_{{\WW,\alpha}}$ can also be used to define incentive and powerset
  network models.
The former can be defined by 
  $u_s(T) = \sum_{v\in T} \PPR_{{\WW,\alpha}}[s,v]$,
   for $s\in V, T\subset V$ and $s\not\in T$.
The latter can be defined by
$\V{\theta}_{\WW}(S,T) = \frac{\sum_{u\in S,v\in T} \PPR_{{\WW,\alpha}}[u,v]}{|S|}$
for $S,T\subseteq V$. 
$\V{\theta}_{\WW}(S,T)$ measures the {\em rate} of PageRank contribution
   from $S$ to $T$.

\subsection{Multifaceted Approaches to Network Analysis: Some Basic Questions}

We will now conclude this section with a few basic questions,
  aiming to study how structural concepts in one network model
  can inspire structural concepts in other network models. 
A broad view of network data
  will enable us to comprehensively examine 
  different facets of network data,
  as each network model brings out different aspects
  of network data.
For examples, 
  the metric model is based on geometry, 
   the preference model is inspired by  social-choice theory \cite{ArrowBook}, 
   the incentive and cooperative models are based on
   game-theoretical and economical principles 
  \cite{NashNonCooperative,NAS50,Shapley},
  the powerset model  is motivated by social influences 
  \cite{DomingosRichardson,RichardsonDomingos,Kempe03},
  while  the graphon \cite{Graphon} is based on graph limits and statistical
 modeling.
We hope that addressing questions below will help us to 
  gain comprehensive and comparative understanding of 
  these models and the network structures/aspects
    that these models may reveal.
We believe that multifaceted and multimodal approaches to network analysis
  will become increasingly more essential 
  for studying major subjects in network science.

\begin{itemize}
\item 
How should we formulate  
  {\em personalized centrality measures} with respect
  to other commonly-used network centrality measures
  \cite{NewmanBook,Bonacich1987power,Freeman,FreemanBetweenness,BorgattiCentrality,BorgattiEverett,EverettBorgattiGroup,BonacichGroupCentrality,Faust,Katz,BavelasCloseness,SabidussiCloseness,PercolationCentrality,Donetti2005Entangled,ShapleyValueForCentrality1,ShapleyValueForCentrality2}?
Can they be used to define meaningful
   centrality-conforming Markov chains?
\item 
How should we define centrality measures and 
  personalized ranking systems
  for general incentive or powerset networks?
How should we define personalized Shapley 
  value for cooperative games?
How should we define weighted networks from cooperative/incentive/powerset
  models?

\item 
What are natural Markov chains associated with the probabilistic graphical models
\cite{KollerFriedman}?
How should we define centrality and clusterability for 
  this important class of network models that 
  are central to statistical machine learning?

\item What constitutes a community in a probabilistic graphical model?
  What constitutes a community in a
   cooperative,  incentive, preference, and powerset network?
How should we capture network similarity 
  in these models?
How should we integrate them if they represents different facets of network data?

\item 
How should we evaluate different clusterability measures and their
  usefulness to community identification 
  or clustering? 
For example, PageRank conductance and PageRank clusterability are
  two different subset functions, but the latter applies to directed networks.
How should we define clusterability-conforming centrality
  or centrality-forming  clusterability?

\item What are limitations of Markovian worldview of various network models?
What are other unified worldview models for multifaceted network data?

\item What is the fundamental difference between ``directed''
  and ``undirected'' networks in various models?
 

\item How should we model networks with non-homogeneous nodes and edge types?
\end{itemize}

More broadly, the objective is to build 
  a systematic algorithmic framework for understanding
  multifaceted network data, particular given that many natural network 
  models are highly theoretical in that their 
  complete-information profiles have exponential dimensionality
  in $|V|$.
In practice, they must be succinctly defined.
The algorithmic network framework consists of the complex and challenging
 tasks of integrating  sparse and 
  succinctly-represented multifaceted network data $N = (V, F_1,...,F_k)$ 
  into an effective worldview $(V,W)$ 
  based on which, one can effectively build succinctly-represented
  underlying models for network facets,
  analyzing the interplay between network facets, and 
  identify network solutions that are consistent
   with the comprehensive network data/models.
What is a general model 
  for specifying multifaceted network data?
How should we formulate the problem
  of {\em network composition} for multifaceted network data?

\section{To Jirka}

The sparsity, richness, and ubiquitousness
 of multifaceted networks data
 make them wonderful subjects
 for mathematical and algorithmic studies.
Network science has truly become a ``universal discipline,'' 
  with its multidisciplinary roots and  interdisciplinary presence.
However, it is a fundamental and conceptually challenging task to
  understand network data, due to the vast network phenomena.

\begin{quote}
{\em The holy grail of network science is to understand 
 the network essence that underlies the observed 
sparse-and-multifaceted network data.
}
\end{quote}

We need an analog of the concept of range space, which 
  provides a united worldview of a 
  family of diverse problems that are fundamental in statistical 
  machine learning, geometric approximation, and data analysis.
I wish that I had a chance to discuss with you about 
  the mathematics of networks --- beyond just the geometry of graphs ---
  and to learn from your  brilliant insights
  into  {\em the essence of networks}.
You and your mathematical depth and clarity will be greatly missed, Jirka.

\bibliographystyle{abbrv}
\bibliography{scalable}
\end{document}